\documentclass[12pt]{elsarticle}

\makeatletter
\def\ps@pprintTitle{%
 \let\@oddhead\@empty
 \let\@evenhead\@empty
 \def\@oddfoot{}%
 \let\@evenfoot\@oddfoot}
\makeatother

\usepackage[reqno]{amsmath}
\usepackage{amsthm,amsfonts,amssymb}

\usepackage[inner=30mm,outer=25mm,top=30mm,bottom=30mm]{geometry}

\usepackage{graphicx,color}

\usepackage{soul}

\newtheorem{Cor}{Corollary}
\newtheorem{Lem}{Lemma}
\newtheorem{Def}{Definition}
\newtheorem{Thm}{Theorem}

\newtheorem{Ex}{Example}

\newcommand{\Z}{{\mathbb Z}}
\newcommand{\ZZ}{\Z_{p^r}}
\newcommand{\ZZN}{\Z_{p^r}^n}
\newcommand{\ZZD}{\Z_{p^r}[D]}
\newcommand{\ZZND}{\Z_{p^r}^n[D]}

\newcommand{\C}{{\mathcal C}}

\newcommand{\wt}{\mathrm{wt}}

\newcommand{\ie}{\emph{i.e.}}

\newcommand{\KER}{\mbox{\rm{Ker}}\,}

\newcommand{\rank}{\mbox{rank}\,}

\usepackage{amssymb}

\makeindex

\begin{document}

\begin{frontmatter}

\title{List decoding of Convolutional Codes over integer residue rings}
\author{Julia Lieb, Diego Napp and Raquel Pinto}

\address{Julia Lieb: Department of Mathematics, University of Zurich, Switzerland, e-mail: julia.lieb@math.uzh.ch.}

\address{Diego Napp: Department of Mathematics, University of Alicante, Spain, e-mail: diego.napp@ua.es\\
%
}

\address{Raquel Pinto: CIDMA - Center for Research and Development in Mathematics and Applications, Department of Mathematics, University of Aveiro, Portugal, e-mail: raquel@ua.pt
%
}

\begin{abstract}
A convolutional code $\C$ over $\ZZ[D]$ is a $\ZZ[D]$-submodule of $\ZZN[D]$ where $\ZZ[D]$ stands for the ring of polynomials with coefficients in $\ZZ$. In this paper, we study the list decoding problem of these codes when the transmission is performed over an erasure channel, that is, we study how much information one can recover from a  codeword $w\in \C$ when some of its coefficients have been erased. We do that using the $p$-adic expansion of $w$ and particular representations of the parity-check polynomial matrix of the code. From these matrix polynomial representations we recursively select certain equations that $w$ must satisfy and have only coefficients in the field $p^{r-1}\ZZ$. We exploit the natural block Toeplitz structure of the sliding parity-check matrix to derive a step by step methodology to obtain a list of possible codewords for a given corrupted codeword $w$, that is, a list with the closest codewords to $w$. 
\end{abstract}

\begin{keyword}
Convolutional codes, finite rings, erasure channel

\end{keyword}

\end{frontmatter}

\section{Introduction}


Convolutional codes form a fundamental class of linear codes that are widely used in applications (see  also the related notion of sequential cellular automata \cite{CaFu17}) . They are typically described by means of a generator matrix, which is a polynomial matrix with coefficients in a finite field or a finite ring, depending on the application. Yet, the mathematical theory of convolutional codes over finite fields is much developed and has produced many sophisticated classes of codes. On the other hand, very little is known about concrete optimal constructions of convolutional codes over finite rings. In any case, the decoding of these codes is, in general, not easy. Probably the most prominent decoding algorithm is the Viterbi algorithm but its use is limited as its complexity grows exponentially with the size of the memory of the code. However, in \cite{to12} it was shown that the decoding of convolutional codes over finite fields requires only linear algebra when they are used over the erasure channel, \ie, when the positions of the errors are known. Despite the fact that convolutional codes that possess optimal erasure correcting capabilities require large finite fields, the results in \cite{to12} allow to implement these codes in many practical situations and therefore attracted the interest of many researchers, see for instance \cite{ma15a} and references therein.


%

Following this thread of research and aiming to extend these results over finite fields to the context of finite rings, we consider in this paper convolutional codes $\C$ over $\ZZ[D]$ and study the erasure correcting capabilities of these codes over the erasure channel. Convolutional codes over $\ZZ [D]$ were introduced by Massey et al. in \cite{massey89} and recently several relevant results were presented \cite{Interlando97,NaPiTo17,OuNaPiTo19,el13}.

In particular, our goal is to retrieve as much information as possible from the received corrupted vector. The decoder proposed in this work is a maximum likelihood algebraic decoder and follows succinctly two main steps. Firstly, it searches for unique decoding, \ie, when there exists a unique most likely word, then, the decoder outputs such a word. When this is not possible the algorithm performs a list decoding algorithm, \ie, it computes a complete list of the most likely codewords for a given corrupted codeword.

For this problem, we shall use the parity-check matrix $H(D)$ of $\C$ in a particular form. Then, the number of independent columns of specific submatrices of $H(D)$ will determine the size of the list of possible codewords in the algorithm. Considering the erasures as unknowns to-be-determined, the decoding  problem treated here amounts to solving a system of linear equations over $\ZZ$. The idea we used in this work is to multiply a selected subset of these equations by a power of $p$ in such a way that we obtain equations with coefficients in  the field $p^{r-1}\ZZ$, isomorphic to $\mathbb{Z}_{p}$, and therefore we can easily compute the unknowns. Once we know  some of the coefficients we can select another set of equations and apply the same ideas to these equations to recover another set of erased symbols. In this way we develop a systematic recursive procedure to recover all possible erasures that could have occurred. This, in turn, provides with a list with the closest codewords to the received information vector.


%
The outline of this paper is as follows. In Section 2, we present basic results on convolutional codes over the finite ring $\ZZ$, in particular about their parity-check matrices, which are important for decoding over the erasure channel. In Section 3, we present our erasure decoding algorithm for convolutional codes over $\ZZ$ and illustrate it with an example. Finally, we conclude with analyzing the complexity of our algorithm in Section 4.

\section{Preliminary results}

In this section we present the elementary background required in the paper. Let $\ZZD$ denote the ring of polynomials with coefficients in $\ZZ$ and let $\mathcal{A} =\{0,1,2, \dots, p-1\}$ be the set of \textbf{digits}. We say that $v(D)$ has \textbf{order} $s$, denoted by ord$(v(D))=s$, if $p^{s-1}v(D) \neq 0$ and $p^{s-1}v(D) \in p^{r-1}\ZZ^\ell[D]$. Every element in $v(D)\in \ZZ^\ell[D] $ admits a unique $p$-adic expansion as $v(D)= a_0(D) + a_1(D) p + \cdots + a_{r-1}(D)p^{r-1}$, with $a_i(D)\in \mathcal{A}[D]$, ord$(a_i(D))=r-1$ and $i=0,1,\dots, r-1$. We shall extensively use that $p^{r-1}\ZZ$ is isomorphic to the field $\Z_p$ in our algorithms.

\begin{Def} \cite{QuNaPiTo17,fo70,el13}
A {\bf convolutional code $\mathcal{C}$} of length $n$ is a $\ZZ[D]$-submodule of  $\ZZN[D]$. A polynomial matrix $ G(D) \in \ZZ^{k \times n}[D]$ such that
\begin{eqnarray*}
{\mathcal C} & =  \left\{ G(D)^T u(D) \in \ZZN[D] :\, u(D) \in \Z^{ k}_{p^r}[D]\right\}
\end{eqnarray*}
is called {\bf generator matrix} of the code.
\end{Def}

A polynomial matrix $H(D)$ is a {\bf parity-check matrix} (or
syndrome former) of a convolutional code $\C$ if $\C= \KER_{\mathbb Z_{p^r}[D] } H(D)$, \ie ,
for every $w(D)\in \ZZN[D]$,

\begin{equation}\label{eq:01}
w(D) \in \C \Leftrightarrow H(D)w(D)=0.
\end{equation}

\begin{Def}
Let $R$ be a ring with identity and $U(D)\in R[D]^{n\times n}$. Then $U(D)$ is called {\bf unimodular} if there exists $V(D)\in R[D]^{n\times n}$ such that $U(D)V(D)=V(D)U(D)=I_n$.
\end{Def}

For any matrix $A$ with entries in $\mathbb Z_{p^r}$ or $\mathbb Z_{p^r}[D]$, we denote by $[A]_p$ the (componentwise) projection of $A$ into $\mathbb Z_p$.

\begin{Lem}\cite{ku07}\label{1}
Let $U(D)\in\mathbb Z_{p^r}[D]^{n\times n}$. Then $U(D)$ is unimodular (over $\mathbb Z_{p^r}[D]$) if and only if $[U(D)]_p$ is unimodular (over $\mathbb Z_p[D]$).
\end{Lem}

\begin{Def}
Let $\mathbb F$ be a finite field.
A polynomial matrix $P(D)\in\mathbb F[D]^{k\times n}$ with $k\leq n$ is {\bf left
prime} if in all factorizations
$P(D) = \Delta(D)\bar{P}(D)$ with $\Delta(D)\in\mathbb F[D]^{k\times k}$ and $\bar{P}(D)\in\mathbb F[D]^{k\times n}$,
the left factor $\Delta(D)$ is unimodular.
\end{Def}

\begin{Lem}\label{2}\cite{ka80}
Let $P(D)\in\mathbb Z_p[D]^{k\times n}$ with $k\leq n$. Then, the following conditions are equivalent:
\begin{itemize}
\item[(i)]$P(D)$ is left prime.
\item[(ii)]There exists $N(D)\in\mathbb Z_p[D]^{(n-k)\times k}$ such that $\begin{pmatrix}P(D)\\ N(D)\end{pmatrix}$ is unimodular.
\item[(iii)] $u(D)P(D)\in\mathbb Z^n_{p^r}[D]\Rightarrow u(D)\in\mathbb Z^k_{p^r}[D]$, for all $u(D)\in\mathbb Z^k_{p^r}(D)$, where $\mathbb Z_{p^r}(D)$ denotes the ring of rational functions with coefficients in $\mathbb Z_{p^r}$.
\end{itemize}   
\end{Lem}

The next result follows immediately from Lemma \ref{1} and Lemma \ref{2}.

\begin{Cor}\label{cor}
Let $P(D)\in\mathbb Z_{p^r}[D]^{k\times n}$ with $k\leq n$. Then $[P(D)]_p$ is left prime over $\mathbb Z_p[D]$ if and only if there exists $N(D)\in\mathbb Z_{p^r}[D]^{(n-k)\times k}$ such that $\begin{pmatrix}P(D)\\ N(D)\end{pmatrix}$ is unimodular.
\end{Cor}

If $\C$ is a convolutional code that admits a parity-check matrix, then a parity-check matrix of $\C$ can be constructed as follows (for more details see \cite{OuNaPiTo19}): Let $\hat{G}(D)\in\mathbb Z_{p^r}[D]^{\hat{k}\times n}$ be a generator matrix of $\mathcal{C}$ and consider $$\hat{\C}=\{\hat{G}(D)^\top u(D):\ u(D)\in\mathbb Z_{p^r}((D))^{\hat{k}}\}$$
where $\ZZ((D))$ denotes the \index{Ring of Laurent series over $\ZZ$}{\bf ring of Laurent series} over $\ZZ$, \ie, the set of elements of the form
$$
a(D) = \sum_{i=-\infty}^{+\infty}{a_iD^i}
$$
where the coefficients $a_i$ are in $\ZZ$ and only finitely coefficients with negative indices may
be nonzero. Note that $\C=\hat{\C}\cap\mathbb Z_{p^r}[D]^n$. Then, there exists
\begin{align}\label{*}
G(D)=\left[
       \begin{array}{c}
         G_0(D) \\
        pG_1(D) \\
        \vdots \\
       p^{r-1}G_{r-1}(D)
       \end{array}
    \right],\ G_i(D)\in\mathbb Z_{p^r}[D]^{k_i\times n}\ \text{for}\ i=0,\hdots,r-1
\end{align}
with $\mathcal{G}=\left[\begin{matrix}G_0(D)\\ G_1(D)\\ \vdots\\ G_{r-1}(D)\end{matrix}\right]_p$ full row rank over $\mathbb Z_p[D]$ such that
$$\hat{\C}=\{G(D)^\top u(D):\ u(D)\in\mathbb Z_{p^r}((D))^k\},$$
where $k=k_0+k_1+\cdots+k_{r-1}$. A procedure to obtain $G(D)$ in this particular form is described in \cite[Theorem 1]{OuNaPiTo19}.
A parity-check matrix of $\hat{C}$ can be determined as follows Consider the generator matrix $G(D)$ defined in \eqref{*} and let $N(D)\in\mathbb Z_{p^r}[D]^{(n-k)\times n}$ such that
such that
$$\begin{pmatrix}
G_0(D)\\G_1(D)\\ \vdots\\G_{r-1}\\N(D)
\end{pmatrix}$$
is unimodular. Then, there exists $H_i(D)\in\mathbb Z_{p^r}[D]^{l_i\times n}$ where $l_0=n-k$ and $l_i=k_{r-i}$, $i=1,2,\hdots,r-1$, and $L(D)\in\mathbb Z_{p^r}[D]^{k_0\times n}$ such that
\begin{align}\label{b}
\begin{pmatrix}
L(D)\\ H_{r-1}(D)\\H_{r-2}(D)\\ \vdots\\H_1(D)\\ H_0(D)
\end{pmatrix}
[G_0(D)^\top\ G_1(D)^\top\ \hdots\ G_{r-1}(D)^\top\ N(D)^\top]=P(D)
\end{align}
for some
$$P(D)=\left[\begin{matrix}
\gamma_1(D) & & &\\ & \gamma_2(D) & & \\ & & \ddots &\\ & & & \gamma_n(D)
\end{matrix}\right],
$$
where $\gamma_i(D)$ are nonzero polynomials in $\mathbb Z_{p^r}[D]$. Then,
 $$H(D)=\left[
       \begin{array}{c}
         H_0(D) \\
        pH_1(D) \\
        \vdots \\
       p^{r-1}H_{r-1}(D)
       \end{array}
    \right],\ H_i(D)\in\mathbb Z_{p^r}[D]^{l_i\times n},\ l_i=k_{r-i},\ \text{for}\ i=1,\hdots,r-1,\ l_0=n-k$$
with $\mathcal{H}=\left[\begin{matrix}H_0(D)\\ H_1(D)\\ \vdots\\ H_{r-1}(D)\end{matrix}\right]_p$ full row rank over $\mathbb Z_p[D]$, such that
$$
w(D)\in \hat{\C} \Leftrightarrow H(D)w(D)=0.
$$

More details can be found in \cite{OuNaPiTo19}. Then,
$$\hat{\C}=\{w(D)\in\mathbb Z_{p^r}((D))^n:\ H(D)w(D)=0\},$$
and consequently,
$$\C=\KER_{\mathbb Z_{p^r}[D]} H(D),$$
which means that $H(D)$ is a parity-check matrix of $\C$.

However, not all convolutional codes admit a parity-check matrix as the following example shows.

\begin{Ex}\label{ce}
Let $G(D)=\left[\begin{matrix}1+D & 1+D & 1+D\\ 3 & 3 & 0\end{matrix}\right]\in\mathbb Z_9[D]^{2\times 3}$ and $\C$ be the convolutional code with generator matrix $G(D)$. Let us assume that $\C$ admits a parity-check matrix $H(D)$. Then, since $w(D)=[1+D\ \ 1+D\ \ 1+D]^\top\in\C$, it follows that $H(D)w(D)=0$. Then, $(1+D)H(D)\begin{pmatrix} 1\\ 1\\ 1\end{pmatrix}=0$ and therefore, $H(D)\begin{pmatrix} 1\\ 1\\ 1\end{pmatrix}=0$. This means that $\begin{pmatrix} 1\\ 1\\ 1\end{pmatrix}\in\C$, which is not true.
\end{Ex}

However, in the case that a convolutional code does not admit a parity-check matrix, the procedure above constructs a matrix $H(D)$ such that $\C\subset \KER_{\mathbb Z_{p^r}[D]} H(D)$ because $\C\subset\hat{\C}$.

If a convolutional code $\C$ admits a parity-check matrix, it is called {\bf observable} or {\bf non-catastrophic}. The characterization of the non-catastrophic convolutional codes over $\mathbb Z_{p^r}[D]$ is still an open problem. However, for a certain class of convolutional codes, the following lemma characterizes the non-catastrophic codes.

\begin{Lem}\label{3}
Let $\C$ be a convolutional code of length $n$ that admits a generator matrix $G(D)$ of the form
\begin{align}G(D)=\left[
       \begin{array}{c}
         G_0(D) \\
        pG_1(D) \\
        \vdots \\
       p^{r-1}G_{r-1}(D)
       \end{array}
    \right],\ G_i(D)\in\mathbb Z_{p^r}[D]^{k_i\times n}\ \text{for}\ i=0,\hdots,r-1
\end{align}
such that
\begin{align}
\mathcal{G}=\left[\begin{matrix}G_0(D)\\ G_1(D)\\ \vdots\\ G_{r-1}(D)\end{matrix}\right]_p
\end{align}
is full row rank over $\mathbb Z_p[D]$. Then, $\C$ admits a parity-check matrix if and only if $\mathcal{G}$ is left prime over $\mathbb Z_p[D]$.
\end{Lem}

\begin{proof}
Take $k=k_0+k_1+\cdots+k_{r-1}$ and let us assume that $G(D)$ is such that $\mathcal{G}$ is left prime over $\mathbb Z_p[D]$. Then, by Corollary \ref{cor}, there exists $N(D)\in\mathbb Z_{p^r}[D]^{(n-k)\times n}$ such that
$$\begin{pmatrix}
G_0(D)\\G_1(D)\\ \vdots\\G_{r-1}\\N(D)
\end{pmatrix}$$
is unimodular. Then
\begin{align}\label{b}
\begin{pmatrix}
L(D)\\ H_{r-1}(D)\\H_{r-2}(D)\\ \vdots\\H_1(D)\\ H_0(D)
\end{pmatrix}
[G_0(D)^\top\ G_1(D)^\top\ \hdots\ G_{r-1}(D)^\top\ N(D)^\top]=I
\end{align}
for some $H_i(D)\in\mathbb Z_{p^r}[D]^{l_i\times n}$ where $l_0=n-k$ and $l_i=k_{r-i}$, $i=1,2,\hdots,r-1$, and $L(D)\in\mathbb Z_{p^r}[D]^{k_0\times n}$. Let us define
\begin{align*}
H(D)=\begin{pmatrix}
H_0(D)\\ pH_1(D)\\ \vdots\\ p^{r-1}H_{r-1}(D).
\end{pmatrix}
\end{align*}
By Corollary \ref{cor} we have that
\begin{align*}
\mathcal{H}=\left[\begin{matrix}H_0(D)\\H_1(D)\\ \vdots\\H_{r-1}(D)\end{matrix}\right]_p
\end{align*}
is left prime over $\mathbb Z_p[D]$. Expression \eqref{b} also implies that
\begin{align}\label{s}
H(D)G(D)^\top=0.
\end{align}
Let us prove that $\mathcal{C}=\KER_{\mathbb Z_{p^r}[D]} H(D)$. Expression \eqref{s} shows that $\mathcal{C}\subset\KER_{\mathbb Z_{p^r}[D]} H(D)$. To prove the other inclusion, let us consider that $w(D)\in\KER_{\mathbb Z_{p^r}[D]} H(D)$, i.e., $w(D)\in \mathbb Z_{p^r}[D]$ is such that
$$\begin{pmatrix}
H_0(D)\\ pH_1(D)\\ \vdots\\ p^{r-1}H_{r-1}(D)
\end{pmatrix}w(D)=0.$$
Then
$$H_i(D)w(D)=p^{r-i}v_i(D)$$
for some $v_i(D)\in\mathbb Z_{p^r}[D]^{l_i}$, $i=0,1,\hdots,r-1$. Since, by \eqref{b},
$$w(D)=[G_0(D)^\top\ G_1(D)^\top\ \cdots\ G_{r-1}(D)^\top\ N(D)^\top]\left[\begin{matrix}L(D)\\ H_{r-1}(D)\\H_{r-2}(D)\\ \vdots\\H_1(D)\\H_0(D)\end{matrix}\right]w(D),$$
it follows that $$w(D)=[G_0(D)^\top\ G_1(D)^\top\ \cdots\ G_{r-1}(D)^\top\ N(D)^\top]\left[\begin{matrix}L(D)w(d)\\ pv_{r-1}(D)\\p^2v_{r-2}(D)\\ \vdots\\p^{r-1}v_1(D)\\0\end{matrix}\right]=G(D)^\top\left[\begin{matrix}L(D)w(D)\\ v_{r-1}(D)\\v_{r-2}(D)\\ \vdots\\v_1(D)\end{matrix}\right],$$
which means that $w(D)\in\mathcal{C}$. Thus, we conclude that $\mathcal{C}=\KER_{\mathbb Z_{p^r}[D]} H(D)$, \ie, $\mathcal{C}$ admits a parity-check matrix.\\
To show the converse, let us assume that $\C$ has a parity-check matrix $H(D)$ and let $\bar{w}(D)\in\mathbb Z^n_{p}[D]$  be such that
$$\bar{w}(D)=\mathcal{G}^\top\bar{u}(D),$$
for some $\bar{u}(D)\in\mathbb Z^k_{p}(D)$. Then,
$$p^{r-1}w(D)=\mathcal{G}^\top p^{r-1}u(D),$$
where $w(D)\in\mathbb Z^n_{p^r}[D]$ and $u(D)\in\mathbb Z^k_{p^r}(D)$ are such that $[w(D)]_p=\bar{w}(D)$ and $[u(D)]_p=\bar{u}(D)$. Therfore,
$$p^{r-1}w(D)=G(D)^\top u_1(D)$$
where
$$p^{r-1}u(D)=\left[\begin{matrix}
I_{k_0} & & &\\ & pI_{k_1} & &\\ & & \ddots &\\ & & & p^{r-1}I_{k_{r-1}}
\end{matrix}\right]u_1(D).
$$
Consequently,
$$H(D)p^{r-1}w(D)=H(D)G(D)^\top u_1(D)=0,$$
which means that $p^{r-1}w(D)\in\C$. Therfore $u_1(D)\in\mathbb Z^k_{p^r}[D]$ and hence, $\bar{u}(D)\in\mathbb Z^k_{p^r}[D]$. Thus, we conclude by Lemma \ref{2} that $\mathcal{G}$ is left prime over $\mathbb Z_p[D]$.

\end{proof}



It is well-known that kernel representations are useful to detect errors introduced during transmission. If a word $w(D)$ is received
after channel transmission, the existence of errors is checked by simple multiplication by $H(D)$: if $H(D)w(D)=0$, it is assumed that
no errors occurred. As Example \ref{ce} and Lemma \ref{3} show, not all convolutional codes defined in $\mathbb Z_{p^r}[D]$ admit a parity-check matrix.
Nevertheless we showed that there always exists a matrix $H(D)$ such that $\C \subset \ker H(D)$, and then we still can make use of $H(D)$ to decode when the transmission occurs over the erasure channel. For simplicity, we will also call this matrix a parity-check matrix of $\C$. In an erasure channel a codeword can only have erasures (\ie, we know the positions of the part of the codeword that is missing or erased) but no errors occur.
In fact, if one considers the erasures as indeterminates, $H(D)w(D)=0$ give rise to a system of linear equations.
Solving this system amounts to decoding the received word $w(D)$, as we explain in detail in the next section.

The associated {\bf truncated sliding parity-check matrix} of $H(D)= \displaystyle \sum_{i=0}^{\nu} H^i D^i,
$ is
\begin{equation}\label{eq:033}
H^c_ j=\left[
\begin{array}{cccc}
H^0 & & & \\
H^1 & H^0 & & \\
\vdots & & \ddots & \\
 H^j & H^{j-1} & \cdots & H^0
\end{array}
\right]
\end{equation}
with  $H^j=0$ for $j>\nu$. As any codeword $w(D)$ of $\C$ satisfies $H(D)w(D)=0$, if $w(D)=\sum_{i \in \mathbb N_0} w^i D^i$, we have that, for all $j\geq 0$, $  \sum_{i=0}^{j} H^i w^{j-i}=0$, \ie,
\begin{equation}\label{eq:00}
\left[
\begin{array}{cccc}
H^0 & & & \\
H^1 & H^0 & & \\
\vdots & & \ddots & \\
 H^j & H^{j-1} & \cdots & H^0
\end{array}
\right]
\left[
\begin{array}{c}
w^0  \\
w^1 \\
\vdots  \\
w^j
\end{array}
\right]
=0.
\end{equation}

Two of the main notions of minimum distance of convolutional codes are the free distance and the column distance. Given $w(D)=\sum_{i \in \mathbb N_0} w^i D^i$, we define its \textbf{Hamming weight} as
$$ \wt(w(D)) = \sum_{j \in \mathbb{N}} \wt(w^j)$$
where $\wt(w^j)$ is the number of nonzero elements of $w^j$.

Given an $ (n,k) $ convolutional code $ \mathcal{C} \subseteq \ZZ^n[D] $, we define its \textbf{free distance} as
\begin{equation*}
d_{free}(\mathcal{C}) = \min \{ \wt(w(D)) \ : \ w(D) \in \mathcal{C} \textrm{ and } w(D) \neq 0 \} .
\end{equation*}

The free distance gives the correction capability of a convolutional code when considering whole codewords. In other words, there is no maximum degree $ j $ for a codeword considered by the free distance. In this work we shall focus on the sliding-window erasure correction capabilities of $\C$ within a time interval and this will be determined by the {\bf $j$-th column distance} of $\C$, for $j\in\mathbb N_0$, which is defined as follows.
\begin{eqnarray}\label{eq:CD}
  \nonumber  \mbox{$d^c_j$}(\mathcal{C}) &=&   \min \{ \wt((w^0, w^1, \ldots, w^j)) \ : \ \sum_{i \in \mathbb{N}_0} w^i D^i \in \mathcal{C}, w^0 \neq 0 \} \\
 & \stackrel{*}{=} & \left\lbrace \wt((w^0, w^1, \ldots, w^j)) \ : \ (w^0, w^1, \ldots, w^j) \text{ satisfies (\ref{eq:00}) and } w^0 \neq 0 \right\rbrace
\end{eqnarray}
where the equality $ * $ holds for convolutional codes that have a parity-check matrix. Next, we present two preliminary results.

%

\begin{Lem}\cite{mcdonald84}\label{lem:mccoy}
  Let $Ax=b$ with $A \in \ZZ^{a \times s}$ and $b \in \ZZ^{a}$ be a linear system of equations in $x$. Suppose this system has a solution. Then, the solution is unique if and only if  $[A]_p$ is full column rank or equivalently, if the McCoy \footnote{The McCoy rank of a matrix is the largest size of a minor that is an invertible element in the ring, $\mathcal{A}\setminus \{0\}$ in our case.} rank of $A$ is $s$.
\end{Lem}

Note that, opposed to the field case,  a set of vectors in $\ZZ$ can be linearly dependent but none of them is in the $\ZZ$-span of the others. The following result states the erasure correcting capability of a convolutional code in terms of its column distance.

\begin{Lem}\label{lem:column_dist}
Let $\C=\KER_{\mathbb Z_{p^r}[D]} H(D)$ and $j\in \mathbb{N}$. The following statements are equivalent:

\begin{enumerate}
  \item the column distance $d^c_j(\C)=d$;
  \item if $ (w^0, w^1, \ldots, w^j) $ contains up to $d-1$ erasures then $w^0$ can be recovered and there exist $d$ erasures that make it impossible to recover $w^0$.
  \item all sets of $d-1$ columns of $H_j^c$ that contain at least one of the first $n$ columns of $H_j^c$ are linearly independent and there exists a set of $d$ columns  of $H_j^c$ that contains at least one of the first $n$ columns of $H_j^c$ and is linearly dependent;
\end{enumerate}
If these equivalent statements hold, then the following holds:
\begin{itemize}
  \item[(4)] none of the first $n$ columns of $[H_j^c]_p $ is contained in the $\Z_p$-span of any other $d-2$ columns of $[H_j^c]_p$.
\end{itemize}

\end{Lem}

\begin{proof}
The equivalence of 1., 2. and 3. can be shown exactly in the same way as for finite fields (see \cite{gl03}). The last statement $(4)$ follows from the fact that if we have a linear combination of columns of $[H_j^c]_p $ in $\Z_p$, we readily obtain a linear combinations of $H_j^c$ over $\ZZ$ (just multiply the coefficients of the linear combination from $\Z_p$ by $p^{r-1}$).
\end{proof}

We notice that statement $(4)$ of Lemma \ref{lem:column_dist} does not imply the others as we show in the following example.
\begin{Ex}
  Consider $\C= \KER_{\Z_9[D]} H(D)$ over $\mathbb{Z}_{9}$ where $H(D)=H^0 + H^1 D$ with
  $$
  H^0=\left(
        \begin{array}{ccc}
          1 & 0 & 3 \\
          0 & 1 & 3 \\
        \end{array}
      \right) \text{ and }
   H^1=\left(
        \begin{array}{ccc}
          0 & 1 & 1 \\
          1 & 0 & 1 \\
        \end{array}
      \right) .
  $$
Then, none of the first 3 columns of $[H^c_1]_3$ is a linear combination of $1=3-2$ of the remaining columns of $[H^c_1]_3$. Hence, $(4)$ is fulfilled for $d=3$. However, it is easy to see that
 $d^c_1(\mathcal{C})= 2$ (consider the truncated codeword $(w^0,w^1)= (003 \ 002)$), i.e., the first statement of Lemma \ref{lem:column_dist} is not fulfilled for $d=3$.
\end{Ex}

\section{A decoding algorithm for erasures}\label{decoding_algorithm}

In this section we state the problem using the notation presented in the previous section and then propose an efficient decoding algorithm to solve it. More concretely, we aim to recover erasures that may occur during the transmission of the information over an erasure channel using convolutional codes $\C\subset\ZZND$. We derive a constructive step by step decoding algorithm to compute a minimal list with the closest codewords to the received vector. This is equivalent to solving a certain system of linear equations in $\ZZ$. 

Suppose that $w(D)=\sum_{i \in \mathbb N_0} w^i D^i \in \C$ is sent and assume that we have correctly received all coefficients up to an instant $i-1$ and some of the components of $w^i$ are erased. The decoder tries to recover  $w^i$ up to a given instant $i+T$ and if this is not possible it outputs a list with the closest vectors at time instant $i+T$. The parameter $T$ is called the delay constraint and represents the maximum delay the receiver can tolerate to retrieve $w^i$, see \cite{ba15b,Martinian2007}. For the sake of simplicity it will be assumed that $T\leq \nu$, where $\nu$ is the degree of the parity-check matrix $H(D)$ of $\C$. The system of equations that involve $w^i$ up to time instant $i+T$ is

\begin{eqnarray}\label{eq0}
  \left[
  \begin{matrix}
    H^{\nu} &  H^{\nu-1} &  \cdots & H^{0}  & & \\
    & H^{\nu} &  H^{\nu-1} &  \cdots & H^{0}  & & \\
    & & \ddots  & \vdots & \vdots & \ddots & \\
  &  &  &  H^{\nu} &  H^{\nu-1} &  \cdots & H^{0}   \\
  \end{matrix}\right]
  \left[
  \begin{matrix}
    w^{i-\nu } \\
    w^{i-\nu +1} \\
    \vdots \\
    w^i \\
    \vdots \\
    w^{i+T}
  \end{matrix}
  \right]={0}.
\end{eqnarray}


We can take the columns of the matrix in (\ref{eq0}) that correspond to the coefficients of the erased elements to be the coefficients of a new system. With the remaining columns we can compute the independent terms, denoted by $ b^i$.

We regard the erasures as to-be-determined variables and denote for $i\in\mathbb N_0$ by $\widetilde w^i$ the subvector of $w^i$ that corresponds to the positions of the erasures. Similarly, denote by $\widetilde H^j_{i}$ the matrix consisting of the columns of $H^j$ with indices corresponding to the erased positions in $w^i$.
Then, we obtain the following system of linear equations
\begin{equation}\label{eq:03}
\left[
\begin{array}{cccc}
\widetilde H^0_{i} & & & \\
\widetilde H^{1}_{i} & \widetilde H^0_{i+1} & & \\
\vdots & & \ddots & \\
\widetilde H^T_{i} & \widetilde H^{T-1}_{i+1} & \cdots & \widetilde H^0_{i+T}
\end{array}
\right]
\left[
\begin{array}{c}
\widetilde w^{i}  \\
\widetilde w^{i+1} \\
\vdots  \\
\widetilde w^{i+T}
\end{array}
\right]=
\left[
\begin{array}{c}
 b^i  \\
 b^{i+1} \\
\vdots  \\
 b^{i+T}
\end{array}
\right].
\end{equation}


Hence, the problem of decoding is equivalent to solving the system of linear equations described in (\ref{eq:03}).

For this algorithm we consider a parity-check matrix $H(D)$ of the form

\begin{equation}\label{eq:standard_form}
H(D)=\left[
\begin{array}{c}
H_0(D)  \\
pH_1 (D) \\
\vdots  \\
p^{r-1}H_{r-1}(D)
\end{array}
\right] \text{ with } \left[
\begin{array}{c}
H_0(D)  \\
H_1 (D) \\
\vdots  \\
H_{r-1}(D)
\end{array}
\right] \text{ full row rank. }
\end{equation}

Hence, it readily follows that one can rewrite (\ref{eq:03}), after appropriate row permutations, as

\begin{equation}\label{eq:04}
\left[
\begin{array}{cccc}
\widetilde  H^0_{i,0} & & & \\
p \widetilde  H^0_{i,1} & & & \\
\vdots & & & \\
p^{r-1} \widetilde H^0_{i,r-1} & & & \\
\widetilde H^1_{i,0} & \widetilde H^0_{i+1,0} & & \\
p\widetilde  H^1_{i,1} & p\widetilde  H^0_{i+1,1} & & \\
\vdots & \vdots & & \\
p^{r-1}\widetilde  H^1_{i,r-1} & p^{r-1}\widetilde  H^0_{i+1,r-1} & & \\
\vdots & & \ddots & \\
\widetilde  H^T_{i,0} & \widetilde  H^{T-1}_{i+1,0} & \cdots & \widetilde  H^0_{i+T,0} \\
p \widetilde  H^T_{i,1} & p \widetilde  H^{T-1}_{i+1,1} & \cdots & p \widetilde  H^0_{i+T,1} \\
\vdots & \vdots & & \vdots \\
p^{r-1} \widetilde  H^T_{i,r-1} & p^{r-1} \widetilde  H^{T-1}_{i+1,r-1} & \cdots & p^{r-1} \widetilde  H^0_{i+T,r-1}
\end{array}
\right]
\left[
\begin{array}{c}
\widetilde w^{i}  \\
\widetilde w^{i+1} \\
\vdots  \\
\widetilde w^{i+T}
\end{array}
\right]=
\left[
\begin{array}{c}
 b^i_0 \\
p b^i_1 \\
\vdots \\
p^{r-1}b^i_{r-1} \\
 b^{i+1}_0 \\
p b^{i+1}_1 \\
\vdots \\
p^{r-1}b^{i+1}_{r-1} \\
\vdots  \\
 b^{i+T}_0 \\
p b^{i+T}_1 \\
\vdots \\
p^{r-1}b^{i+T}_{r-1} \\
\end{array}
\right],
\end{equation}
with the property that the rows of the matrices $\widetilde H^0_{i,t}, [\widetilde H^1_{i,t} \; \widetilde H^0_{i+1,t}], \dots, [\widetilde H^T_{i,t} \;\widetilde H^{T-1}_{i+1,t} \cdots \widetilde H^0_{i+T,t}]$ have  order $r$
, for $t=0,1, \dots, r-1$. Note that the number of nonzero rows of each block in the decomposition (\ref{eq:04}) will depend on the erasure pattern.\\
Denote by $e_s$ the size of $\widetilde w^s$, for $s\in \{i, i+1, \dots, i+T  \}$. \\


\textbf{List decoding}:\\
We aim to compute all possible solutions of (\ref{eq:04}). To this end we define the following matrix for all $0 \leq t \leq r-1$,

\begin{equation}\label{eq:06}
\widetilde{ \mathcal{H}}^c_t= \left[
\begin{array}{cccc}
\widetilde  H^0_{i,0} & & & \\
 \widetilde  H^0_{i,1} & & & \\
\vdots & & & \\
 \widetilde H^0_{i,r-t-1} & & & \\
\widetilde H^1_{i,0} & \widetilde H^0_{i+1,0} & & \\
\widetilde  H^1_{i,1} & \widetilde  H^0_{i+1,1} & & \\
\vdots & \vdots & & \\
\widetilde  H^1_{i,r-t-1} & \widetilde  H^0_{i+1,r-t-1} & & \\
\vdots & & \ddots & \\
\widetilde  H^T_{i,0} & \widetilde  H^{T-1}_{i+1,0} & \cdots & \widetilde  H^0_{i+T,0} \\
\widetilde  H^T_{i,1} &  \widetilde  H^{T-1}_{i+1,1} & \cdots &  \widetilde  H^0_{i+T,1} \\
\vdots & \vdots & & \vdots \\
\widetilde  H^T_{i,r-t-1} & \widetilde  H^{T-1}_{i+1,r-t-1} & \cdots & \widetilde  H^0_{i+T,r-t-1}
\end{array}
\right],
\end{equation}

and write

\begin{equation}\label{eq:adic_w}
\left[
\begin{array}{c}
\widetilde w^i  \\
\widetilde w^{i+1} \\
\vdots  \\
\widetilde w^{i+T}
\end{array}
\right]=\left[
\begin{array}{c}
 w^{i}_0  \\
w^{i+1}_0 \\
\vdots  \\
w^{i+T}_0
\end{array}
\right]+p \left[
\begin{array}{c}
 w^{i}_1  \\
 w^{i+1}_1 \\
\vdots  \\
w^{i+T}_1
\end{array}
\right]+\cdots+p^{r-1}\left[
\begin{array}{c}
 w^{i}_{r-1}  \\
 w^{i+1}_{r-1} \\
\vdots  \\
 w^{i+T}_{r-1}
\end{array}
\right],
\end{equation}
where $w^j_{t}$ has entries in $\mathcal{A}_{p}=\{0,1,\dots,p-1\}$, for all $j\in\{i,i+1,\dots,i+T\}$ and $t\in \{0,1,\dots,r-1\}$. We aim at computing the maximum number of coefficients $w^j_{t}$ in (\ref{eq:adic_w}).

\begin{description}
  \item[{\bf \underline{Step 1}}:] Find the solution $(\widehat w^{i}_0, \widehat w^{i+1}_0, \dots , \widehat w^{i+T}_0)$ of the system
\begin{align}\label{eq:10}
\left[
\begin{array}{cccc}
\widetilde H^0_{i,0} & & & \\
\widetilde H^0_{i,1} & & & \\
\vdots & & & \\
\widetilde H^0_{i,r-1} & & & \\
\widetilde H^1_{i,0} & \widetilde H^0_{i+1,0} & & \\
\widetilde H^1_{i,1} & \widetilde H^0_{i+1,1} & & \\
\vdots & \vdots & & \\
\widetilde H^1_{i,r-1} & \widetilde H^0_{i+1,r-1} & & \\
\vdots & & \ddots & \\
\widetilde H^T_{i,0} & \widetilde H^{T-1}_{i+1,0} & \cdots & \widetilde H^0_{i+T,0} \\
\widetilde H^T_{i,1} & \widetilde H^{T-1}_{i+1,1} & \cdots & \widetilde H^0_{i+T,1} \\
\vdots & \vdots & & \vdots \\
\widetilde H^T_{i,r-1} & \widetilde H^{T-1}_{i+1,r-1} & \cdots & \widetilde H^0_{i+T,r-1}
\end{array}
\right]_p
\left[
\begin{array}{c}
\widehat w^i_0  \\
\widehat w^{i+1}_0 \\
\vdots  \\
\widehat w^{i+T}_0
\end{array}
\right]=
\left[
\begin{array}{c}
 b^i_0 \\
 b^{i}_1 \\
\vdots \\
b^i_{r-1} \\
 b^{i+1}_0 \\
 b^{i+1}_1 \\
\vdots \\
b^{i+1}_{r-1} \\
\vdots  \\
 b^{i+T}_0 \\
 b^{i+T}_1 \\
\vdots \\
b^{i+T}_{r-1} \\
\end{array}
\right]_{p},
\end{align}
\begin{equation*}
\hspace*{-4.5cm} \underbrace{\hspace*{5.5cm}}_{ [\widetilde{\mathcal{H}}^c_0]_p}
\end{equation*}
over the field $\Z_p$. Let $e= \sum_{s=i}^{i+T}e_s$. Then, the ``integer" part of the set of solutions, \ie, \ the vector $\left[
\begin{array}{c}
 w^{i}_0  \\
w^{i+1}_0 \\
\vdots  \\
w^{i+T}_0
\end{array}
\right]$ in (\ref{eq:adic_w}), is given by:
\[
 \cal S_\text{$0$} = \left \{
 \left[
\begin{array}{c}
w^i_0  \\
w^{i+1}_0 \\
\vdots  \\
w^{i+T}_0
\end{array}
\right] \in \mathcal{A}^\text{$e$} \ : \  \left[
\begin{array}{c}
w^i_0  \\
w^{i+1}_0 \\
\vdots  \\
w^{i+T}_0
\end{array}
\right]_{\text{$p$}}=\left[
\begin{array}{c}
\widehat w^i_0  \\
\widehat w^{i+1}_0 \\
\vdots  \\
\widehat w^{i+T}_0
\end{array}
\right]  \text{ with  } \left[
\begin{array}{c}
\widehat w^i_0  \\
\widehat w^{i+1}_0 \\
\vdots  \\
\widehat w^{i+T}_0
\end{array}
\right]  \text{ satisfying (\ref{eq:10})  }
\right \}.
\]

It is straightforward to see that the size of $\cal S_\text{$0$}$ is given by
$$
|\mathcal{S}_0| = p^{e- \rank [\widetilde{ \mathcal{H}}^c_0]_p}.
$$

To compute the remaining vectors, if necessary, in the $p$-adic decomposition of (\ref{eq:adic_w}) we recursively apply the following algorithm in the next step.

  \item[{\bf \underline{Step 2}}:] Let $b_{s,0}^{j}=b_s^j$, $j=i,i+1,\dots,i+T$, $s=0,1, \dots, r-1$.

  For $t=1, \dots,r-1$ do
  \begin{enumerate}
    \item For $j=i,i+1,\dots,i+T$, consider the solutions $\left[
\begin{array}{c}
w^i_{t-1}  \\
w^{i+1}_{t-1} \\
\vdots  \\
w^{i+T}_{t-1}
\end{array}
\right]\in S_{t-1}$ and define
    \begin{align*}
\begin{bmatrix}
\widehat{b}^j_{0,t} \\
\widehat{b}^j_{1,t} \\
\vdots\\
\widehat{b}^j_{r-t-1,t}
\end{bmatrix}
& =
\begin{bmatrix}
p^{t-1}b^j_{0,t-1} \\
p^{t}{b}^j_{1,t-1} \\
\vdots\\
p^{r-2}{b}^j_{r-t-1,t-1}
\end{bmatrix}-
\begin{bmatrix}
p^{t-1}\widetilde H^{j-i}_{i,0} & \cdots & p^{t-1}\widetilde H^0_{j,0} \\
p^{t}\widetilde H^{j-i}_{i,1} & \cdots & p^{t}\widetilde H^0_{j,1} \\
\vdots & \vdots & \vdots \\
p^{r-2}\widetilde H^{j-i}_{i,r-t-1} & \cdots & p^{r-2}\widetilde H^0_{j,r-t-1}
\end{bmatrix}
\left[
\begin{array}{c}
w^i_{t-1}  \\
w^{i+1}_{t-1} \\
\vdots  \\
w^{j}_{t-1}
\end{array}
\right]
 \end{align*}


    \item For $j=i,i+1,\dots,i+T$, and $\ell=0,1,\dots, r-t-1$, compute \emph{one} $b_{\ell,t}^{i}$ such that
\begin{align*}
\widehat{b}_{\ell,t}^{j}=p^{t+\ell}b_{\ell,t}^{j}
\end{align*}
    \item Solve the system of linear equations
\begin{align}\label{eq:11}
[\widetilde{ \mathcal{H}}^c_t]_p
\left[
\begin{array}{c}
\widehat w^i_t  \\
\widehat w^{i+1}_t \\
\vdots  \\
\widehat w^{i+T}_t
\end{array}
\right]=
\left[
\begin{array}{c}
 b^i_{t,0} \\
 b^i_{t,1} \\
\vdots \\
b^0_{t,r-t-1} \\
 b^{i+1}_{t,0} \\
 b^{i+1}_{t,1} \\
\vdots \\
b^{i+1}_{t,r-t-1} \\
\vdots  \\
 b^{i+T}_{t,0} \\
 b^{i+T}_{t,1} \\
\vdots \\
b^{i+T}_{t,r-t-1} \\
\end{array}
\right]_{p},
\end{align}
over $\Z_p$ and let
\[
 \cal S_\text{$t$} = \left \{
 \left[
\begin{array}{c}
w^i_t  \\
w^{i+1}_t \\
\vdots  \\
w^{i+T}_t
\end{array}
\right]  \ : \  \left[
\begin{array}{c}
w^i_t  \\
w^{i+1}_t \\
\vdots  \\
w^{i+T}_t
\end{array}
\right]_{\text{$p$}}=\left[
\begin{array}{c}
\widehat w^i_t  \\
\widehat w^{i+1}_t \\
\vdots  \\
\widehat w^{i+T}_t
\end{array}
\right]  \text{ with  } \left[
\begin{array}{c}
\widehat w^i_t  \\
\widehat w^{i+1}_t \\
\vdots  \\
\widehat w^{i+T}_t
\end{array}
\right]  \text{ satisfying (\ref{eq:11})  }
\right \}.
\]
\end{enumerate}
\end{description}

\underline{{\em Output data}}:

$$
 \left \{
\left[
\begin{array}{c}
 w^i_0  \\
w^{i+1}_0 \\
\vdots  \\
w^{i+T}_0
\end{array}
\right]+p \left[
\begin{array}{c}
 w^{i}_1  \\
 w^{i+1}_1 \\
\vdots  \\
w^{i+T}_1
\end{array}
\right]+\cdots+p^{r-1}\left[
\begin{array}{c}
 w^i_{r-1}  \\
 w^{i+1}_{r-1} \\
\vdots  \\
 w^{i+T}_{r-1}
\end{array}
\right] \ : \ \left[
\begin{array}{c}
w^i_t  \\
w^{i+1}_t \\
\vdots  \\
w^{i+T}_t
\end{array}
\right] \in S_t, \ t=0,1,\dots ,r-1 \right \}.
$$

The size of the list decoding is
$$
\displaystyle \prod^{r-1}_{t=0} |S_t|.
$$
where each $|S_t|$ is given by
\begin{equation}\label{eq:13}
|S_t|= p^{e-  \rank [\widetilde{ \mathcal{H}}^c_t]_p}
\end{equation}

Note that Steps 1 and 2 deal with systems of linear equations over fields. The fact that these steps yield the set of all solution follows from \cite[Theorem 3]{Elif_Marisa_Raquel_Diego}.

\begin{Ex}
Let $\mathcal{C}\subset\mathbb Z_8[D]$ be the convolutional code with parity-check matrix $H(D)=H^0+H^1D+H^2D^2\in\mathbb Z_8[D]$ where $$H^0=\left[\begin{array}{ccccc} 1 & 1 & 1 & 1 & 1 \\ 0 & 0 & 2 & 0 & 2\\ 4 & 4 & 0 & 4 & 4\end{array}\right],\ H^1=\left[\begin{array}{ccccc} 1 & 2 & 0 &0 & 0 \\ 0 & 0 & 0 & 2 & 4\\ 4 & 0 & 4 & 4 & 0\end{array}\right],\ H^2=\left[\begin{array}{ccccc} 3 & 5 & 7 & 0 & 0 \\ 0 & 0 & 0 & 0 & 2\\ 0 & 0 & 0 & 4 & 0\end{array}\right].$$
It is easy to check that $w(D)=w^0+w^1D+w^2D^2+w^3D^3$ with
$$w^0=[5,5,0,6,0],\ \ w^1=[6,6,4,3,6],\ \ w^2=[2,1,1,2,0],\ \ w^3=[2,6,4,0,0]$$
is a codeword of $\mathcal{C}$. Assume that one receives $$w^0=[5,w^{0,1},w^{0,2},6,w^{0,3}], w^1=[6,6,4,w^{1,1},6], w^2=[2,1,w^{2,1}, w^{2,2}, w^{2,3}], w^3=[2,w^{3,1},4,0,0]$$ where $w^{0,1},w^{0,2},w^{0,3},w^{1,1},w^{2,1},w^{2,2},w^{2,3},w^{3,1}$ are erasures. Let the delay constraint for the decoding be $T=2$. To firstly recover $w^0$ we start our list decoding algorithm. One has
$$\widetilde{ \mathcal{H}}_0^c=\left[\begin{array}{ccccccc} 1 & 1 & 1 & 0 & 0 &0 & 0\\
 0 & 1 & 1 & 0 & 0 & 0 & 0\\ 1 & 0 & 1 & 0 & 0 & 0 & 0\\ 2 & 0 & 0 & 1 & 0 & 0 & 0\\0 & 0 & 1 & 0 & 0  &0 & 0\\ 0 & 1& 0 & 1 & 0 & 0 & 0\\ 5 & 7 & 0 & 0 & 1 & 1 & 1\\ 0 & 0 & 1 & 1 & 1 &0 & 1\\ 0 & 0 & 0 & 1 & 0 & 1 & 1\end{array}\right]\qquad\text{and}\qquad[\widetilde{ \mathcal{H}}_0^c]_2=\left[\begin{array}{ccccccc} 1 & 1 & 1 & 0 & 0 &0 & 0\\
 0 & 1 & 1 & 0 & 0 & 0 & 0\\ 1 & 0 & 1 & 0 & 0 & 0 & 0\\ 0 & 0 & 0 & 1 & 0 & 0 & 0\\0 & 0 & 1 & 0 & 0  &0 & 0\\ 0 & 1& 0 & 1 & 0 & 0 & 0\\ 1 & 1 & 0 & 0 & 1 & 1 & 1\\ 0 & 0 & 1 & 1 & 1 &0 & 1\\ 0 & 0 & 0 & 1 & 0 & 1 & 1\end{array}\right].$$
We write $w^{i,j}=w_{0}^{i,j}+2w_{1}^{i,j}+4w_{2}^{i,j}$ for $i=0,\hdots , 3$ and $j\in\{1,2,3\}$.
Solving the linear system
$$[\widetilde{ \mathcal{H}}_0^c]_2\cdot[w^{0,1}_{0}, w^{0,2}_{0},w^{0,3}_{0},w^{1,1}_{0},w^{2,1}_{0},w^{2,2}_{0},w^{2,3}_{0}]^{\top}=
[5,0,1,5,0,1,4,0,1]^{\top}_2=[1,0,1,1,0,1,0,0,1]^{\top}$$




over $\mathbb Z_2$ gives the (unique) solution
$$[w^{0,1}_{0}, w^{0,2}_{0},w^{0,3}_{0},w^{1,1}_{0},w^{2,1}_{0},w^{2,2}_{0},w^{2,3}_{0}]=
[1,0,0,1,1,0,0],$$ i.e., $S_0=\{[1,0,0,1,1,0,0]\}$. Note that $|S_0|=p^{e-  \rank [\widetilde{ \mathcal{H}}^c_0]_p}=p^{7-7}=1$.\\
Then, in step 2.1 and step 2.2 of the algorithm, one computes
\begin{eqnarray*}
\begin{pmatrix}
\widehat{b}_{0,1}^0\\ \widehat{b}_{1,1}^0
\end{pmatrix}=\begin{pmatrix}5\\0\end{pmatrix}-\left[\begin{array}{ccc} 1& 1 & 1\\ 0 & 2 & 2\end{array}\right]\begin{pmatrix}1\\0\\0\end{pmatrix}=
\begin{pmatrix}4\\0\end{pmatrix} \quad&\Rightarrow\ \begin{pmatrix} b_{0,1}^0\\ b_{1,1}^0 \end{pmatrix}=
\begin{pmatrix}2\\0\end{pmatrix}\\
\widehat{b}_{0,1}^1=5-\left[\begin{array}{cccc} 2&  0 & 0 & 1\end{array}\right]\begin{pmatrix}1\\0\\0\\1\end{pmatrix}=
\begin{pmatrix}4\\0\end{pmatrix} \quad&\Rightarrow\  b_{0,1}^1 =
1\\
\begin{pmatrix}
\widehat{b}_{0,1}^2\\ \widehat{b}_{1,1}^2
\end{pmatrix}=\begin{pmatrix}4\\0\end{pmatrix}-\left[\begin{array}{ccccccc} 5& 7 & 0 & 0 & 1 & 1 & 1 \\ 0 & 0 & 2 & 2 & 2 & 0 & 2\end{array}\right]\begin{pmatrix}1\\0\\0\\1\\1\\0\\0\end{pmatrix}=
\begin{pmatrix}6\\4\end{pmatrix} \quad&\Rightarrow\ \begin{pmatrix} b_{0,1}^2\\ b_{1,1}^2 \end{pmatrix}=
\begin{pmatrix}3\\1\end{pmatrix}\
\end{eqnarray*}
Afterwards, according to step 2.3, one has to solve the system of linear equations
$$\left[\begin{array}{ccccccc} 1 & 1 & 1 & 0 & 0 &0 & 0\\
 0 & 1 & 1 & 0 & 0 & 0 & 0\\ 0 & 0 & 0 & 1 & 0 & 0 & 0\\ 1 & 1 & 0 & 0 & 1 & 1 & 1\\ 0 & 0 & 1 & 1 & 1 &0 & 1\end{array}\right]\begin{pmatrix}w^{0,1}_{1}\\ w^{0,2}_1\\w^{0,3}_{1}\\w^{1,1}_{1}\\w^{2,1}_{1}\\w^{2,2}_{1}\\w^{2,3}_{1}
 \end{pmatrix}=\begin{pmatrix}0\\0\\1\\1\\1\end{pmatrix}$$
 over $\mathbb Z_2$, which yields $$[w^{0,1}_{1}, w^{0,2}_1, w^{0,3}_{1}, w^{1,1}_{1}, w^{2,1}_{1}, w^{2,2}_{1}, w^{2,3}_{1}]=[0,c_1,
 c_1,1,c_1+c_2,1,c_2]$$ with free parameters $c_1,c_2\in\mathbb Z_2$, i.e., $$S_1=\{[0,c_1,c_1,1,c_1+c_2,1,c_2],\ c_1,c_2\in\mathcal{A}_2\}$$ with $|S_1|=p^{7-\rank [\widetilde{ \mathcal{H}}^c_1]_p}=p^2=4$.\\
 In the last iteration, one computes
 \begin{eqnarray*}
\widehat{b}_{0,2}^0=2\cdot 2-2\cdot\left[\begin{array}{ccc} 1& 1 & 1\end{array}\right]\begin{pmatrix}0\\c_1\\c_1\end{pmatrix}=
4-4c_1\quad&\Rightarrow\ b_{0,2}^0 =1-c_1\\
\widehat{b}_{0,2}^1=2\cdot 1-2\cdot\left[\begin{array}{cccc} 2&  0 & 0 & 1\end{array}\right]\begin{pmatrix}0\\c_1\\c_1\\1\end{pmatrix}=0\quad&\Rightarrow\  b_{0,2}^1 =0\\
\widehat{b}_{0,2}^2=2\cdot 3-2\cdot\left[\begin{array}{ccccccc} 5& 7 & 0 & 0 & 1 & 1 & 1 \end{array}\right]\begin{pmatrix}0\\c_1\\c_1\\1\\c_1+c_2\\1\\c_2\end{pmatrix}=
4-4c_2 \quad&\Rightarrow\ b_{0,2}^2=1-c_2
\end{eqnarray*}
and afterwards solves the system of linear equations
$$\left[\begin{array}{ccccccc} 1 & 1 & 1 & 0 & 0 &0 & 0\\
 0 & 0 & 0 & 1 & 0 & 0 & 0\\ 1 & 1 & 0 & 0 & 1 & 1 & 1\\ \end{array}\right]\begin{pmatrix}w^{0,1}_{2}\\ w^{0,2}_{2}\\w^{0,3}_{2}\\w^{1,1}_{2}\\w^{2,1}_{2}\\w^{2,2}_{2}\\w^{2,3}_{2}
 \end{pmatrix}=\begin{pmatrix}1-c_1\\0\\1-c_2\end{pmatrix}$$
 over $\mathbb Z_2$, which yields
\begin{align*}
 &[w^{0,1}_{2}, w^{0,2}_{2}, w^{0,3}_{2}, w^{1,1}_{2}, w^{2,1}_{2}, w^{2,2}_{2}, w^{2,3}_{2}]\\
&= [1+c_2+c_3+c_4+c_5+c_6,c_3,c_1+c_2+c_4+c_5+c_6,0,c_4,c_5,c_6]
\end{align*}
  with free parameters $c_3,c_4, c_5,c_6\in\mathbb Z_2$, i.e.
  $$S_2=\{[1+c_2+c_3+c_4+c_5+c_6,c_3,c_1+c_2+c_4+c_5+c_6,0,c_4,c_5,c_6],\ c_3,c_4, c_5,c_6\in\mathcal{A}_2\}.$$
with $|S_2|=p^{7-\rank [\widetilde{ \mathcal{H}}^c_2]_p}=p^4=16$.\\
In summary, all solutions for the erased positions are given by
$$\begin{pmatrix}w^{0,1}\\ w^{0,2}\\w^{0,3}\\w^{1,1}\\w^{2,1}\\w^{2,2}\\w^{2,3}
 \end{pmatrix}=\begin{pmatrix}1\\0\\0\\1\\1\\0\\0\end{pmatrix}+
 2\begin{pmatrix}0\\c_1\\c_1\\1\\c_1+c_2\\1\\c_2\end{pmatrix}+4\begin{pmatrix}1+c_2+c_3+c_4+c_5+c_6\\c_3\\
 c_1+c_2+c_4+c_5+c_6\\0\\c_4\\c_5\\c_6\end{pmatrix}$$
with $c_1,c_2,c_3,c_4,c_5,c_6\in\mathcal{A}_2$, i.e., $$|S|=p^{0+2+4}=p^6=64.$$
 Note that for $c_1=c_2=c_3=c_3=c_4=c_5=c_6=0$, one gets the solution that leads to the original codeword we started with. Because of the constraint $T=2$, the vector $w^3$ is not recovered yet. However, since all other erasures are recovered, the remaining erasure $w^{3,1}$ can now easily be recovered.
\end{Ex}

Of course the smaller the size of the output the better. Obviously, this holds if $\rank [\widetilde{ \mathcal{H}}^c_t]_p$ is maximal.

\section{Complexity Analysis}

In this section, we briefly want to analyze the complexity of the presented decoding algorithm. As we work with the projections of elements from the finite ring $\mathbb Z_{p^r}$ in the finite field $\mathbb Z_p$, we can state the computational effort in terms of the number of necessary field operations in $\mathbb Z_p$.

\begin{Thm}
Denote by $e$ the maximal number of erasures that occur in a window of size $(T+1)n$. The number of necessary field operations in $\mathbb Z_p$ for our list decoding algorithm is 
\begin{eqnarray*}O(re^2((n-k)(T+1))^{0.8}) &  \text{if}\ e\leq (n-k)(T+1)\\ O(re^{0.8}((n-k)(T+1))^{2}) & \text{if}\ e> (n-k)(T+1).\end{eqnarray*}
\end{Thm}

\begin{proof}
The step of the algorithm that is relevant for the complexity of the whole algorithm is to solve the system of linear equations in \eqref{eq:11}. This linear system has at most $(n-k)(T+1)$ equations and at most $e$ unknowns. It follows from \cite{st} that the number of field operations that is needed to do that is 
\begin{eqnarray*}O(e^2((n-k)(T+1))^{0.8}) & \text{if}\ e\leq (n-k)(T+1)\\ O(e^{0.8}((n-k)(T+1))^{2}) & \text{if}\ e> (n-k)(T+1).\end{eqnarray*}
The theorem follows from the fact that we have to solve \eqref{eq:11} for $t=0,\hdots,r-1$, what gives us the factor $r$.
\end{proof}


\section*{Acknowledgements}
Julia Lieb acknowledges the support of the German Research Foundation grant LI 3101/1-1 and of the Swiss National Science Foundation grant n. 188430.  Diego Napp is partially supported by Ministerio de Ciencia e Innovación via the grant with ref. PID2019-108668GB-I00.
Raquel Pinto is supported by The Center for Research and Development
in Mathematics and Applications (CIDMA) through the Portuguese
Foundation for Science and Technology
(FCT - Funda\c{c}\~{a}o para a Ci\^{e}ncia e a Tecnologia),
references UIDB/04106/2020 and UIDP/04106/2020.
\bibliographystyle{plain}
\bibliography{biblio_com_tudo}
\end{document}